\documentclass[conference]{IEEEtran}
\ifCLASSINFOpdf
\else
\fi

\usepackage{amsthm}
\usepackage{amsmath}
\usepackage{graphicx,enumitem}
\usepackage{subfigure}
\usepackage[compatible]{algpseudocode}
\usepackage{amssymb}
\usepackage{amsfonts}
\usepackage{algorithm}
\usepackage{algcompatible}
\usepackage{tabularx}
\usepackage{lipsum}

\usepackage[utf8]{inputenc}

\newtheorem{lem}{Lemma}
\usepackage{url}
\usepackage{balance}
\usepackage{epstopdf}
\usepackage{array}
\usepackage{tabularx}
\makeatletter
\newcommand*{\rom}[1]{\expandafter\@slowromancap\romannumeral #1@}
\makeatother

\begin{document}
%

\title{On Maximizing Task Throughput in IoT-enabled 5G Networks under Latency and Bandwidth Constraints\vspace{-10pt}}

\author{\IEEEauthorblockN{Ajay Pratap, Ragini Gupta, Venkata Sriram Siddhardh Nadendla and Sajal K. Das}
\IEEEauthorblockA{Department of Computer Science, Missouri University of Science and Technology, Rolla, MO 65409 USA\\
E-mail:\{ajaypratapf, rg5rv, nadendla, sdas\}@mst.edu\vspace{-17pt}}}
\maketitle
\begin{abstract}
Fog computing in 5G networks has played a significant role in increasing the number of users in a given network. However, Internet-of-Things (IoT) has driven system designers towards designing heterogeneous networks to support diverse demands (tasks with different priority values) with different latency and data rate constraints. In this paper, our goal is to maximize the total number of tasks served by a heterogeneous network, labeled \emph{task throughput}, in the presence of data rate and latency constraints and device preferences regarding computational needs. Since our original problem is intractable, we propose an efficient solution based on graph-coloring techniques. We demonstrate the effectiveness of our proposed algorithm using numerical results, real-world experiments on a laboratory test-bed and comparing with the state-of-the-art algorithm.

\end{abstract}
\vspace{-10pt}
%
\IEEEpeerreviewmaketitle
\section{Introduction \label{sec: Introduction}}
The proliferation of heterogeneous computing devices in everyday life has enabled system designers to implement Internet-of-Things (IoT) and improve services in diverse domains such as healthcare, manufacturing and transportation. Despite the availability of high-performance computing abilities at the base station (BS), the main challenge in designing an IoT network is to provide services seamlessly in the presence of ever-increasing number of edge devices \cite{givehchi2017interoperability}. At the same time, IoT devices are designed to facilitate user mobility which results in limited resources in terms of battery power, bandwidth and computational capacity \cite{samie2016distributed}. Therefore, fog computing has been proposed in 5G networks to reduce the computational load at the BS and support the ever-increasing number of mobile IoT devices \cite{bonomi2012fog}. The IoT devices in fog networks generate tasks via sensing information from a physical phenomenon, and send pre-processed data to a gateway node called \emph{fog access point} (FAP). Upon receiving the data from IoT devices, FAP executes the task and sends the response back to respective IoT devices. 

For example, consider a smart-health IoT network designed to serve stroke patients in a rehabilitation center. While it is necessary to continuously monitor various signals such as blood pressure, heart rate and blood sugar levels in multiple patients, there are other tasks such as fall detection (typically detected using accelerometers, gyroscopes and surveillance cameras) that play a crucial role in the avoidance of accidents during the rehabilitation period. Therefore, tasks such as fall detection take precedence over processing blood sugar readings. At the same time, the latency and bandwidth requirements for video streaming in surveillance cameras are significantly larger than those needed to communicate and process fall detection data. In other words, IoT devices typically generate heterogeneous demands (multi-priority tasks) that require diverse resource requirements (e.g. bandwidth, computational power) in the presence of non-identical latency constraints. In such a scenario, BS should prioritize tasks that need to be served and allocate necessary resources accordingly to different FAPs, via integrating heterogeneous constraints and dynamic network environments \cite{atat2017enabling}. 

\subsubsection*{Related Work}
A few attempts have been made to address similar problems in recent literature. For example, Samie \emph{et al.} have proposed a novel resource management scheme for IoT devices in \cite{samie2016distributed}, where they have reasoned out the need for a discrete number of resources at different stages of operation in the context of a smart health application. Further, authors have proposed a QoS based resource allocation approach for smart-health care application in IoT-enabled networks. Authors have pointed out the necessity of an optimizing resource allocation approach that not only monitors the resource constraints but also keeps latency bound into consideration.  Abedin \emph{et al.} have proposed a matching theory based pairing model for resource sharing in \cite{abedin2015fog} using Irving's stable roommate algorithm. This approach did not consider the latency constraint into the model. On the other hand, Zhang \emph{et al.} have addressed the resource allocation problem using a three-tier solution, which is based on Stackelberg games and matching theory. This approach is not applicable when the IoT devices will have heterogeneous non-uniform latency constraint. In \cite{gu2018joint}, a joint radio and computational resource allocation in the IoT Fog computing model has been studied. Additionally, the authors have proposed student project allocation based matching approach to solving the resource allocation procedure. In \cite{vu2018joint} authors have proposed a joint energy and latency optimization framework for IoT enabled fog access radio. The authors further proposed a knapsack based approach to solve the optimization problem. Latency and non-sharability of limited resources are the main drawbacks of this approach. The authors have assumed the predefined static capacity of each FAP and as soon as the total demand of resources goes higher than the capacity, IoT requests start falling down. Dynamic adaptation and prioritization of different IoT devices over limited available resources is another shortcoming of the existing works.

\subsubsection*{Our Contributions}
In this paper, our goal is to maximize the total number of heterogeneous tasks (a.k.a. task throughput) served by the 5G network in the presence of data rate/latency constraints, along with device preferences regarding computational needs. Given that this problem is NP-Hard, we propose an graph-coloring based algorithm with pseudo-polynomial time complexity, to find an efficient solution. The novelty in our solution approach lies in our system framework where FAP nodes relay the task requests submitted by the IoT devices to the BS, so that the BS can centrally allocate optimal resources. Based on the priority of the tasks and network connectivity, the BS first identifies all the high priority tasks and allocates necessary resources to appropriate FAPs. If there are any residual resources that remain unassigned (or can be reused whenever FAPs are non-interfering with each other), the BS allocates them to serve the low priority tasks. Our solution approach outperforms state-of-the-art algorithm in the literature because it relies on the notion of reuse of resources whenever FAPs are non-interfering with each other. 

The remaining of paper organized as follows. Section \ref{sec: System Model} includes system model. Section \ref{Sec: Problem Formulation} illustrates the problem formulation. Section \ref{Sec: Proposed Solution} presents proposed algorithm. Section \ref{Sec: Performance Study} describes performance measure and Section \ref{Sec: Conclusion} concludes this work. 
\vspace{-15pt}
\section{System Model}\label{sec: System Model}
Consider a network shown in Fig. \ref{Arc}, where $M$ heterogeneous IoT devices request bandwidth and computational resources to BS regarding their respective tasks. Let $\mathbb{I} = \{I_1, \cdots, I_M\}$ denote the different IoT devices in the network. Assume that the $i^{th}$ device $I_i$ generates a task $s_i \in \Psi$, where $\Psi$ represents the set of tasks that the network can execute. 
Assume that there are $K$ fog access points (FAPs) in the network labeled as $\mathbb{F} = \{F_1, \cdots, F_K\}$, which are equipped with CPU cycle rates $\mathbb{C}=\{c_1, \cdots, c_K\}$. In other words, $F_k$ can process a task at the rate of $1/c_k$ computations per unit cycle. In such a case, the goal of the BS is maximize the total number of tasks served, via finding appropriate pairs of IoT devices and FAPs in order to reduce the overall latency in the system, while simultaneously increasing the overall productivity in terms of resource utilization.

Let $\Phi_k$ denote the set of all IoT devices that are within the physical proximity of $F_k$. Note that $\Phi_k$ can also be interpreted as the coverage area of $F_k$. Therefore, it is natural for BS to assign $F_k$ to all the IoT devices within $\Phi_k$ in order to minimize latency. However, it is also possible that the coverage areas of two nearby FAPs can overlap, which leads to interference (consequently, a reduction in the achievable data rate) in the communication between the IoT devices within the overlap region and the corresponding FAPs \cite{peng2016fog}. Furthermore, once the BS matches an IoT device to a FAP, the IoT device shares the task details to the FAP using one or more Physical Resource Blocks (PRBs), which is the smallest unit of communication resource assigned by the BS \cite{etsi2011evolved, zirwas2018sub}. We have considered Orthogonal Frequency Division Multiple Access (OFDMA) model where, a PRB comprises of 180 KHz bandwidth ($\Delta f$) and 0.5 ms time frame \cite{SMC}.


This interplay between IoT devices, FAPs and PRBs can be formally characterized by defining the association within any triplet $(I_i, F_k, n) \in \mathbb{I} \times \mathbb{F} \times \mathbb{N}$ as
\begin{equation}\label{eq1}
    y_{k,i}^n=
    \begin{cases}
      1, & \text{if the $n^{th}$ PRB is assigned to } (I_i, F_k)
      \\[1ex]
      0, & \text{otherwise}.
    \end{cases}
\end{equation}
Interference can occur whenever the following possibilities happen: (i) the same PRB is assigned to two IoT devices $I_i$ and $I_j$ which are within the coverage area of $F_k$, or (ii) the same PRB is assigned to two IoT devices, wherein one of them (say $I_i$) is in coverage areas of both $F_k$ and $F_{k'}$ and the other (say $I_j$) is in the non-overlapping regions of either $F_k$ or $F_{k'}$. In other words, we have the following two conditions:
\begin{enumerate}[label=($C_\arabic*$)]
\item Given $n^{th}$ PRB, we have
$\displaystyle \sum_{i = 1}^M y_{k,i}^n = 1,$
for all $F_k \in \mathbb{F}$.
\item Given $n^{th}$ PRB, we have
$y_{k,i}^n + y_{k',j}^n \leq 1,$
for all $I_i \in \Phi_k \cap \Phi_{k'}$ and $I_j \in \Phi_k - \Phi_{k'}$, for any $F_k$ and $F_{k'}$ in $\mathbb{F}$.
\end{enumerate}
\begin{figure}[!t]
	\centering
	\includegraphics[width=0.38\textwidth]{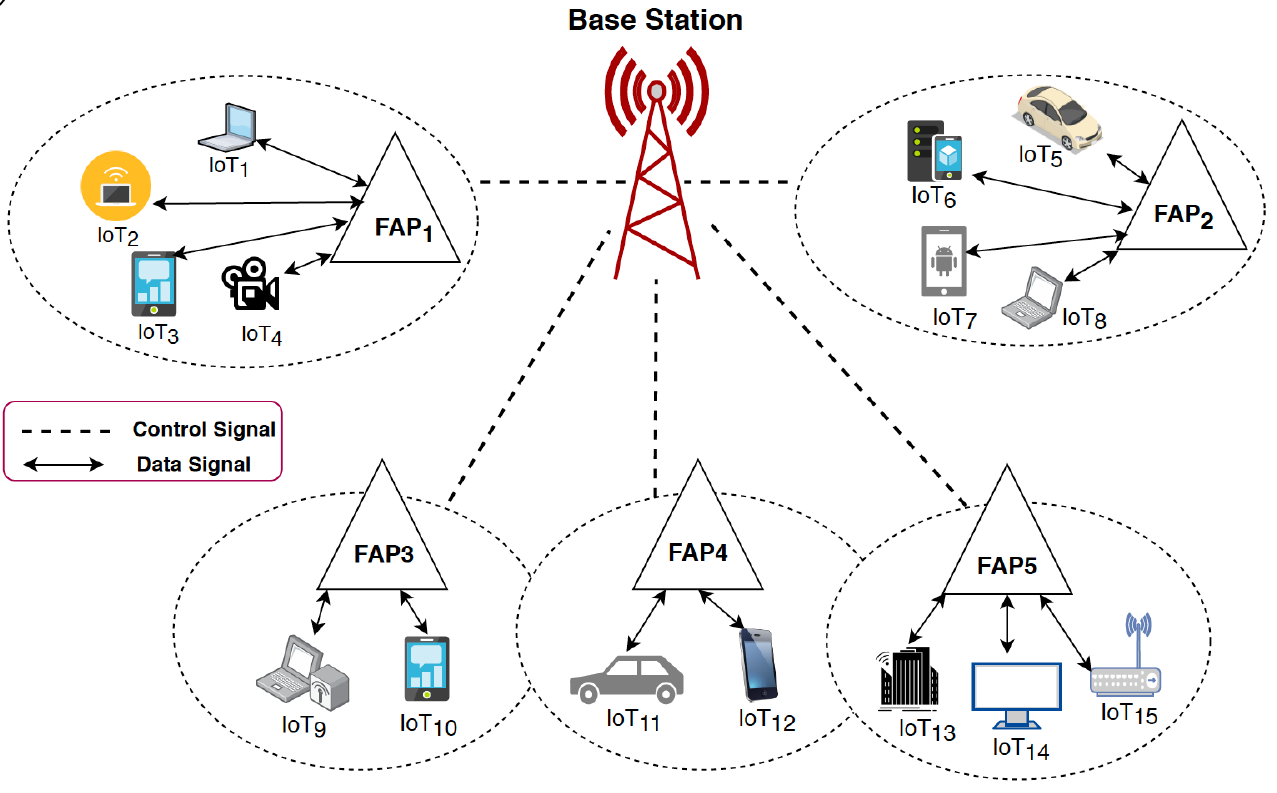}
	\vspace{-2ex}
	\caption{Communication framework for IoT enabled 5G.}
	\vspace{-3ex}
	\label{Arc}
\end{figure}
\subsection{Uplink Latency}
The signal-to-interference-plus-noise ratio (SINR) for a given triplet $(I_i, F_k, n) \in \mathbb{I} \times \mathbb{F} \times \mathbb{N}$ in the uplink channel can be evaluated as, \vspace{-10pt}
\begin{equation}\label{eq4}
\omega_{k,i}^n = \displaystyle \frac{p_i G_{k,i}^n}{\displaystyle \sum_{j \in \mathbb{I}, \ j \neq i} y_{k,j}^n p_{j} G_{k,j}^n + \sigma^2 },
\end{equation}
where $p_i$ is the transmit power employed by $I_i$, $G_{i,k}^n$ is the channel gain between $I_i$ and $F_k$ on the $n^{th}$ PRB, and $\sigma^2$ is the noise variance. However, multiple PRBs are required to increase the data rate for successful communication between any matched pair of IoT devices and FAPs. Let us assume that there are a total of $N$ PRBs $\mathbb{N} = \{ 1, \cdots, N \}$, each with bandwidth $\Delta f$. Let $\mathbb{D}_i \in \mathbb{N}$ be the subset of PRBs assigned to $I_i$ by the BS. Then, the resulting sum-rate in the uplink channel can be computed as,
\begin{equation}\label{eq5}
\Omega_{k,i}(\boldsymbol{y}_{k,i}) = \displaystyle \sum_{n \in \mathbb{D}_i} y_{k,i}^n \cdot \Delta f \cdot \log(1 + \omega _{k,i}^n),
\end{equation}
where $\boldsymbol{y}_{k,i} = \{ y_{k,i}^1, \cdots, y_{k,i}^N \}$ is the vector of PRB allocations made by the BS to the pair $(I_i, F_k)$.

In this paper, we assume that the BS leases the subset of PRBs $\mathbb{D}_i \in \mathbb{N}$ to the pair $(I_i, F_k)$ until the completion of entire process. If $u_{s_i}$ be the size (in terms of bits) of the message that $I_i$ communicates to its assigned FAP, then the total time taken to communicate a task $s_i$ to a FAP is given by,
\begin{equation}\label{eq6}
\alpha_{k,i}(\boldsymbol{y}_{k,i}) = \displaystyle \frac{u_{s_i}}{\Omega_{k,i}(\boldsymbol{y}_{k,i})}.
\end{equation}
\subsection{Execution Latency}
Let $x_{k,i}$ denote the BS's match between the IoT device $I_i$ with the FAP $F_k$, i.e.,\vspace{-10pt} 
\begin{equation}\label{eq7}
x_{k,i}(\boldsymbol{y}_{k,i}) =
\begin{cases}
1, & \text{if } \displaystyle \sum_{n \in \mathbb{N}} y_{k,i}^n \neq 0 
\\
- \infty, & \text{otherwise}.
\end{cases}
\end{equation}

Given that the network can execute any given task in $\Psi$, let us assume that any FAP in the network can execute a specific task $s \in \Psi$ in a total of $\lambda_s$ CPU cycles. Therefore, total time taken to execute a task $s_i$ generated by $I_i$ at $F_k$ is given by,
\begin{equation}\label{eq8}
\beta_{k,i}(\boldsymbol{y}_{k,i}) = \displaystyle | x_{k,i}(\boldsymbol{y}_{k,i}) | \cdot \lambda_{s_i} \cdot \frac{1}{c_k}.
\end{equation}
\subsection{Downlink Latency}
The SINR for a given triplet $(I_i, F_k, n) \in \mathbb{I} \times \mathbb{F} \times \mathbb{N}$ on the downlink channel can be evaluated as,
\begin{equation}\label{eq9}
\theta_{k,i}^n = \displaystyle \frac{\pi_k G_{k,i}^n}{\displaystyle \sum_{k' \in \mathbb{F}, \ k' \neq k} y_{k',i}^n \pi_{k'} G_{k',i}^n + \sigma^2 },
\end{equation}
where $\pi_k$ is the transmit power employed by $F_k$. Since $\mathbb{D}_i \in \mathbb{N}$ denote the subset of PRBs assigned to $I_i$ by the BS, the resulting sum-rate in downlink channel can be computed as,
\begin{equation}\label{eq10}
\Theta_{k,i}(\boldsymbol{y}_{k,i}) = \displaystyle \sum_{n \in \mathbb{D}_i} y_{k,i}^n \cdot \Delta f \cdot \log(1 + \theta_{k,i}^n).
\end{equation}

If $v_{s_i}$ be the size (in terms of bits) of the message that $I_i$ communicates to its assigned FAP. Then, the total time taken to communicate the task $s_i$ to the FAP is given by,
\begin{equation}\label{eq11}
\gamma_{k,i}(\boldsymbol{y}_{k,i}) = \displaystyle \frac{v_{s_i}}{\Theta_{k,i}(\boldsymbol{y}_{k,i})}.
\end{equation}
\subsection{Deadlines and Priorities}

Let $\tau_i$ denote the execution deadline to accomplish a task $s_i$. In other words, we need
\begin{equation}\label{eq12}
\alpha(\boldsymbol{y}_{k,i}) + \beta(\boldsymbol{y}_{k,i}) + \gamma(\boldsymbol{y}_{k,i}) \leq \tau_i.
\end{equation}
We also assume that a given task should at most be assigned to one FAP, in order to avoid computational redundancy within the network. In other words, we need to ensure
\begin{equation}\label{eq13}
\displaystyle \sum_{i \in \mathbb{I}} x_{k,i}(\boldsymbol{y}_{k,i}) \leq 1.
\end{equation}

Furthermore, note that the BS prioritizes its assignment based on the priority values augmented by the IoT devices to their task requests. We denote the priority value declared by $I_i$ regarding its task $s_i$ using a binary variable $w_i \in \{ 0,1 \}$, where $w_i = 0$ corresponds to the lower priority weight and $w_i = 1$ corresponds to the higher priority weight.
\section{Problem Formulation}\label{Sec: Problem Formulation}
Our goal is to maximize the task throughput, i.e. the total number of tasks executed in the network, which is given by,
\begin{equation}\label{eq14}
\eta(Y) = \displaystyle \sum_{k \in \mathbb{F}} \sum_{i \in \mathbb{I}} x_{k,i}(\boldsymbol{y}_{k,i}).
\end{equation}
However, we also want to ensure that the high-priority tasks are given the largest number of PRBs, which is denoted as $d$. The BS can assign the remaining resources to low priority tasks, where each task can at most have $d$ PRBs. This can be formulated as the following problem statement:

\begin{equation}\label{Prob: Original}
\begin{array}{ll}
\displaystyle \max_{Y} & \eta(Y)
\\[2ex]
\text{subject to } & \text{1. } \alpha(\boldsymbol{y}_{k,i}) + \beta(\boldsymbol{y}_{k,i}) + \gamma(\boldsymbol{y}_{k,i}) \leq \tau_i,
\\
& \qquad \qquad \qquad \qquad \qquad \qquad \quad \text{ for all } i \in \mathbb{I} 
\\[1ex]
& \text{2. } \displaystyle \sum_{i \in \mathbb{I}} x_{k,i}(\boldsymbol{y}_{k,i}) \leq 1,
\\[4ex]
& \text{3. } \displaystyle y_{k,i}^n + y_{k',j}^n \leq 1, 
\\
& \qquad \qquad \qquad \text{for all } j \in Int_i, k' \in k \cup Int_k 
\\[2ex]
& \text{4. } \displaystyle y_{k,i}^n \in \{0,1\}, x_{c_k,i} \in \{0,1\}
\\[3ex]
& \text{5. } \displaystyle \sum_{n \in \mathbb{N}} w_i y_{k,i}^n  = d_i, \text{ for all $i \in \mathbb{I}$,}
\\[4ex]
& \text{6. } \displaystyle \sum_{n \in \mathbb{N}} y_{k,i}^n \leq d_i, \text{ for all $i \in \mathbb{I}$,}
\end{array}
\tag{P1}
\end{equation}
where Constraint 5 enforces the BS to allocate exactly $d$ number of PRBs for each high-priority task, i.e., 
\begin{center}
$|\mathbb{D}_i| = d_i$, for all $i \in \mathbb{I}$ such that $w_i = 1$, 
\end{center}
and similarly, Constraint 6 enforces the BS to allocate at most $d$ number of PRBs for all tasks, i.e., 
\begin{center}
$|\mathbb{D}_i| \leq d_i$, for all $i \in \mathbb{I}$. 
\end{center}

Note that the problem is formulated to facilitate resource reuse at the BS, especially when any assignment does not cause any interference to others. Furthermore, note that each of the above stages consists of a non-linear binary program, which is generally a NP-hard problem \cite{bertsimas1997introduction}. To reduce the computational complexity, we propose an efficient algorithm based on graph coloring in the following section.

\section{Proposed Solution}\label{Sec: Proposed Solution} 
Before we present our solution to Problem \eqref{Prob: Original}, we first represent the interference graph construction. Using this graphical representation of interference within the network, we design an algorithmic solution to our resource allocation problem based on graph-coloring approach. 
\subsection{Interference Graph Construction} 

As discussed in above Section \ref{Sec: Problem Formulation}, IoT devices belonging to interfering FAP should not be assigned with the same PRBs. In order to avoid the interference constraint we construct interference graph as follows. We assumed that BS will construct a graph showing the interference scenarios among different FAPs. A FAP considers the other FAP as its neighbour if strength of control packet transmitted from other FAP is more than a predefined threshold \footnote{Defining the appropriate threshold and designing the control packet are out of scope of this paper.}. We assumed that each FAP can identify its neighboring FAPs and they report to BS, and BS maintains a interference graph accordingly as shown in Fig. \ref{Topo}. Let interference graph be $G(V,E)$, where, vertex $V$ represents set of FAPs and edge $e_{k,k'} \in E$ represents interference relation between any two FAPs $F_k$ and $F_{k'}$ such as defined below:     
\begin{equation}\label{eq18}
e_{k,k'} =
\begin{cases}
1, & \text{if FAP $F_k$ is a neighbor FAP
of FAP $F_{k'}$} 
\\
0, & \text{otherwise}.
\end{cases}
\end{equation}
\begin{figure}[h!]
\vspace{-10pt}
	\centering
	\includegraphics[width=0.35\textwidth]{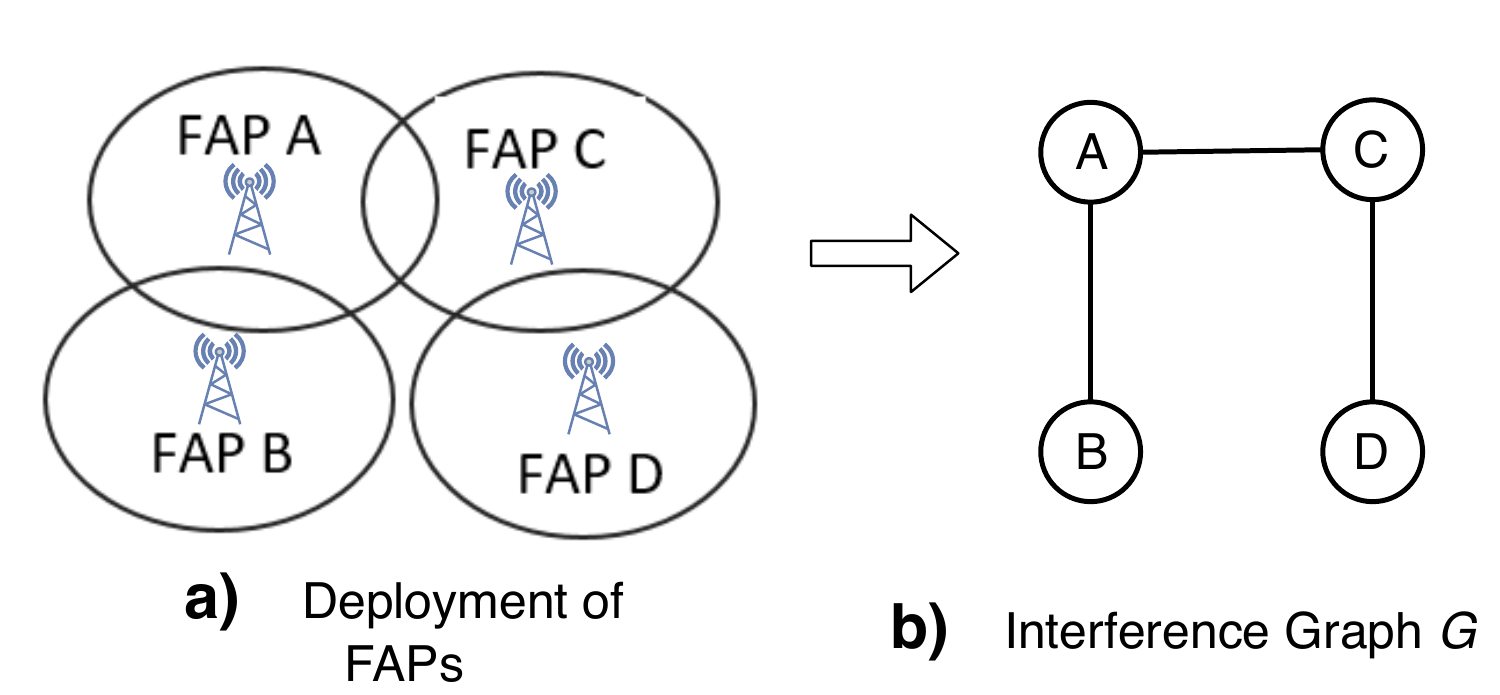}
	\vspace{-2ex}
	\caption{Interference graph.}
	\label{Topo}
	\vspace{-2ex}
\end{figure}
We also assumed that neighboring relation between two FAP is symmetric i.e., $e_{k,k'}=e_{k',k}$. Thus, in order to satisfy the conditions given in Conditions (C1) and (C2), we avoid the same PRB allocation to any two neighboring FAPs. 

\subsection{Proposed Algorithm} We design the graph coloring based resource allocation algorithm in order to maximize the number of high priority IoT devices and at the same time provide the required number of PRBs to low priority IoT devices as given in formulated Problem P1. Let's say the total number of required PRBs of high priority IoT devices as minimum requirement of a FAP whereas, the total number of required PRB's demand including high priority and low priority IoT devices as maximum quota of FAP in-order to achieve the required transmission rate between IoT and respective FAP. On the other hand the execution latency depends on the CPU cycle rate of FAP. For the sake of simplicity, we assume each co-CPU task will be allocated an equal share of the total CPU rate of respective FAP. To fulfill the minimum requirement and maximum quota of FAPs, we consider that BS will first assign the minimum required PRBs of FAPs and then try to fulfill the maximum quota of respective FAPs. Let, $D_k^{min}=\sum_{i \in \mathbb{I}\& w_i=1}d_i$ and $D_k^{max}=\sum_{i \in \mathbb{I}\& w_i=\{0,1\}}d_i$ be the minimum requirement and maximum quota of PRBs of a particular FAP $F_k$, respectively. Consequently, we can write $D_k^{min}=\sum_{i \in \mathbb{I}\& w_i=1}\sum_{n\in \mathbb{N}}y_{k,i}^n$ and $D_k^{max}=\sum_{i \in \mathbb{I}\& w_i=\{0,1\}}\sum_{n\in \mathbb{N}}y_{k,i}^n$ as the minimum requirement and maximum quota of a FAP $k$. Thus, each FAP divided into two dummy FAPs dubbed as steady FAP and elaborate FAP. The steady FAP keeps track of minimum requirement whereas elaborated FAP undertakes the remaining quota of device. Let $k^m$ and $k^q$ represent the steady and elaborated FAP of original FAP $F_k$, respectively. The maximum quota of steady and elaborated FAP can be written as $D_{k^m}^{max}=D_{k}^{min}$ and $D_{k^q}^{max}=D_{k}^{max}-D_{k}^{min}$, respectively. Thus, there is no minimum requirement of all dummy FAPs. If we reserve the $\sum_{k} D_k^{min}$ PRBs for steady FAPs means, minimum requirement of original FAP is fulfilled and the remaining resources $R=N-\sum_{k} D_k^{min}$ will be assigned to the elaborated FAPs. Moreover, to find out how many number of PRBs are required for steady FAPs is not so trivial. For example, let there be 10 FAPs with minimum PRB requirement of 2 and none of them interfere with each other. Thus, instead of reserving 20 resources, we can just keep 2 resources and by reusing it, the FAPs can fulfill their minimum requirement. In order to sort out this problem, let $R$ resources are assigned to elaborated FAPs and $N-R$ resources are fixed for steady FAPs. Thus, we first assign the resources to steady FAPs and then remaining $R$ resources can be assigned to the elaborated FAPs. 

\subsubsection{Minimum Resource Requirement} 
\begin{algorithm}[h!]                    
	\caption{Minimum Resource Requirement Algorithm}\label{Algb}          
	\label{findme}                          
	\begin{algorithmic} [1]                  
		\item[]	 \textbf{Input:} Minimum required resources $D_k^{min}, \forall F_k \in \mathbb{F}$, interference graph $G$, FAP set $\mathbb{F}$, PRB set $\mathbb{N}$, a set $L:=\phi$  
		\item[]	 \textbf{Output:} Number of resources $R$ for elaborated FAPs.
		\STATE \textbf{for} all node $k$ in $G$ do
		\STATE \hspace{0.25cm} Create clique of size $D_k^{min}$.
		\STATE \hspace{0.25cm} Each node of clique $D_k^{min}$ inherits all the edges of $k$ $\in G$ and a new extended interference graph $G'(V',E')$ is generated corresponding to steady FAPs requirements. 
		\STATE \textbf{end for} 
		\STATE \textbf{while} $V'$ is not empty do
		\STATE \hspace{0.25cm} Remove a random node $v_k$ from $V'$
		\STATE \hspace{0.25cm} Assign the lowest index resource $n\in \mathbb{N}$, that is not already assign to any one hop neighboring node of $v_k$.
		\STATE \hspace{0.25cm} $L=L \cup \{n\}$
		\STATE \textbf{end while}
		\STATE $R=\mathbb{N}-{L}$
	\end{algorithmic}
\end{algorithm}
In Algorithm \ref{Algb}, we estimate the number of reserved resources for steady FAPs. If each node requests for one resource then obtaining the number of reserved resources is equivalent to graph coloring problem. Moreover, to verify that a given graph is colorable with $\kappa$ number of colors or not is equivalent to $\kappa$-coloring problem and this is well known as NP-complete problem. As the minimum requirement of each FAP could be more than one PRBs, so we create virtual clique of size $D^{min}$ for each FAP and all nodes in clique share the interference relation of original graph $G$. The new derived graph is considered as $G'(V',E')$. We applied the greedy graph coloring technique to find out the minimum number of required PRBs. In order to do so, we select a random vertex $v_k$ from set $V'$ and assign the smallest indexed PRB that is not allocated to any one hop neighboring nodes of node $v_k$ (lines 5-7). The assigned PRB is updated in set $L$ (line 8). The above process continues for all vertices in $V'$. Now the set $L$ contains all the PRBs required for steady FAPs and $R$ is the remaining set of PRBs that can be assigned to the elaborated FAPs set.

In the following, we first give an algorithm for transformation of  interference graph with minimum requirement into the resource allocation without minimum requirement followed by the resource allocation algorithm. 

\subsubsection{Interference Graph Transformation }
\begin{algorithm}                    
	\caption{Interference Graph Transformation Algorithm}\label{Alg1}          
	\label{findme}                          
	\begin{algorithmic} [1]                  
		\item[]	 \textbf{Input:} Quota of resources $D_k^{min}, D_k^{max}$ \ $\forall F_k \in \mathbb{F}$, interference graph $G(V,E)$, PRB set $\mathbb{N}$

		\item[]	 \textbf{Output:} FAP set $\breve{\mathbb{F}}$, PRB set $\breve{\mathbb{N}}$, maximum quota of FAPs $\breve{D}_k^{max}, \forall k \in \breve{\mathbb{F}}$, derived interference graph $\breve{G}$.  
		
		\STATE $\breve{\mathbb{F}}^s=\phi,\breve{\mathbb{F}}^e=\phi ,\breve{\mathbb{N}}=\mathbb{N}$.		
		\STATE \textbf{for} all $F_k\in \mathbb{F}$ do
		\STATE \hspace{0.25cm} $\breve{\mathbb{F}}^m=\breve{\mathbb{F}}^m \cup \{k^m\} ,\breve{\mathbb{F}}^q= \breve{\mathbb{F}}^q \cup \{k^q\}$
		\STATE \hspace{0.25cm} $D_{k^{m}}^{max}=D_{k}^{min}$, $D_{k^{q}}^{max}=D_{k}^{max}-D_{k}^{min}$. 
		\STATE \textbf{end for}
			\STATE  \textbf{for} all $k\in G$ do
			\STATE \hspace{0.45cm} Add $k^m$ and $k^q$ to $\breve{G}$.
			\STATE \hspace{0.45cm} $k^m$ and $k^q$ inherits all the edges of $k$ in $G$.
			\STATE \hspace{0.45cm} Create an edge between $k^m$ and $k^q$.
		\STATE \textbf{end for}
	\end{algorithmic}
\end{algorithm}

Algorithm \ref{Alg1} shows the procedure for interference graph transformation into study and elaborated nodes. The number of resources $\mathbb{N}$ remain unchanged, list of steady and elaborated FAPs are initialized as null (line 1). The newly generated interference graph $\breve{G}$ contains the steady and elaborated nodes for each FAP in $G$. Steady and elaborated FAPs inherit all the edges of initial interference graph $G$ along with an edge between them because these two nodes can not share the same PRB, in-fact they represent the same FAP (lines 6-10). We use the transformed graph as an input of resource allocation algorithm procedure given in Algorithm \ref{Alg2}.

 \begin{algorithm}[h!]                    
	\caption{Resource Allocation Algorithm}\label{Alg2}          
	\label{findme}                          
	\begin{algorithmic} [1]                  
		\item[]	 \textbf{Input:} Maximum quota of FAPs $\breve{D}_k^{max}, \forall k \in \breve{\mathbb{F}}$, transformed interference graph $\breve{G}$.  
		\item[]	 \textbf{Output:} Set of PRBs allocated nodes i.e., $\breve{\mu}(k)$, $\forall k \in \breve{\mathbb{F}}$.
		
		\STATE $\forall k \in \breve{\mathbb{F}}, \breve{\mu}(k)=\phi$, waiting list $W_k=\phi$. 
		\STATE $\forall n \in \breve{\mathbb{N}}, \breve{\mu}(n)=\phi$, candidate list $A_n=\breve{\mathbb{F}}$.		
		\STATE \textbf {while} $\exists A_n \neq \phi$ do
		\STATE \hspace{0.20cm} \textbf{for} all resource $n$ with $|A_n|>0$ do
		\STATE \hspace{0.40cm} $Z_n:=$ FAPs those satisfy $k\in A_n, \forall k' \in \mu(n), e_{k,k'}=0.$
		\STATE \hspace{0.40cm} Find the maximum weighted independent set on $Z_n$ as $Z_n^{max}$.
		\STATE \hspace{0.20cm} \textbf{end for}
		\STATE \hspace{0.20cm} \textbf{if} $\forall n,Z_n^{max}=\phi$ then 
		\STATE  \hspace{0.40cm} Return $\breve{\mu}$
		\STATE \hspace{0.20cm} \textbf{else}
		\STATE \hspace{0.40cm} \textbf{for} all FAP $k\in Z_n^{max}$ do
		\STATE \hspace{0.60cm} BS sends $msg<k,n>$ 
		\STATE \hspace{0.60cm} $A_n=A_n-\{k\}$
		\STATE \hspace{0.40cm} \textbf{end for}
		\STATE \hspace{0.40cm} Upon deceive $msg<k,n>$ on FAP $k$
		\STATE \hspace{0.60cm} FAP $k$ updates its waiting list $W_k=W_k \cup \{n\}$.
		\STATE \hspace{0.20cm} \textbf{end if}
		\item[]
		\STATE \hspace{0.20cm} \textbf{for} all FAP $k$ for which $W_k\neq \phi$ do
		\STATE \hspace{0.40cm} \textbf{if} $k\in \breve{\mathbb{F}}^m$ then
		\STATE \hspace{0.60cm} FAP $k$ accepts $D_{k^{m}}^{max}$ resources from $W_k \cup \breve{\mu}(k)$ and reset $W_k=\phi$.
		\STATE \hspace{0.40cm} \textbf{else}  
		\STATE \hspace{0.60cm} $count=0$, Index $l=1$
		\STATE \hspace{0.60cm} \textbf{for} all $k \in \breve{\mathbb{F}}^q$ do
		\STATE \hspace{0.75cm}  $W_k=W_k \cup \mu(k), \mu(k)=\phi$
		\STATE \hspace{0.60cm} \textbf{end for}
		\STATE \hspace{0.60cm} \textbf{while} $cnt <  R$ and $\exists k^q, W_{k^q}\neq \phi, |\breve{\mu}(k^q)|< \breve{D}_{k^q}^{max}$ do
		\STATE \hspace{0.75cm} \textbf{if} $ |\mu(k_l^q)|< \breve{D}_{k_l^q}$ and $W_{k_l^q}\neq \phi$, then
		\STATE \hspace{0.85cm} FAP $k_l^q$ selects PRB $n$ from set $W_{k_l^q}$, such that $\breve{\mu}(k_l^q)|=\breve{\mu}(k_l^q)| \cup {n}$. 
		\STATE \hspace{0.85cm} $W_{k_l^q}=W_{k_l^q}-\{n\}$
		\STATE \hspace{0.85cm} $count=count+1$
		\STATE \hspace{0.75cm} \textbf{end if}
		\STATE \hspace{0.60cm} $l=l+1$ 
		\STATE \hspace{0.60cm} \textbf{end while}
		\STATE \hspace{0.40cm} \textbf{end if}
		\STATE \hspace{0.40cm} $W_k=\phi$
		\STATE \hspace{0.20cm} \textbf{end for}
		\STATE \textbf{end while}
	\end{algorithmic}
\end{algorithm}
\subsubsection{Resource Allocation Algorithm}   

Given the input of transformed FAPs and resources, the Algorithm \ref{Alg2} shows the steps of resource allocation procedure. We have introduced different rules for steady and elaborated FAPs in-order to fulfill the minimum and maximum required PRBs. Let $A_n$ be the set of FAP that the BS has not tried for PRB $n$ allocation. In each round BS selects the set of non-interfering FAPs for PRB $n$ (Line 5). In the other-word, BS finds the maximum weighted independent set among FAPs based on priority and $Z_n^{max}$ is the best set of FAPs that BS can send message for PRB $n$ in the current round. If for all the PRBs $n \in \breve{\mathbb{N}}$ the set $Z_n^{max}$ become empty then Algorithm \ref{Alg2} terminates (lines 8-9). The BS sends message $msg<k,n>$ to all the $k \in Z_n^{max}$ for PRB $n$ and update candidat list $A_n$. Upon receiving message $msg<k,n>$ the FAP $k$ adds PRB $n$ into its waiting list (lines 15-16).

We have considered different rules for steady and elaborated FAPs to select the valid PRBs. A steady FAP $k$ for which waiting list $W_k$ is not empty, selects its required $D_{k^{m}}^{max}$ resources, and reset its waiting list (lines 18-20). For the elaborated FAPs, we follow the following method. Let, $count$ be the counter to estimate the number of resources allocated to all elaborated FAPs and $l$ is the index of elaborated FAP initialize to 1 (line 22). We first put all the allocated PRBs of elaborated FAPs into waiting list and then we sequentially check the elaborated FAP into a specific order for PRB allocation. Let index $l$ be the index of FAP that is being considered.  

If the number of allocated PRBs to elaborated FAP becomes $R$ (i.e., the same as maximum PRBs for all elaborated FAPs), or fulfill the maximum quota or empty waiting list, then the algorithm returns to line 18. Otherwise, FAP $k_l^q$ selects the PRB from its waiting list as long as its maximum quota is not fulfilled and its waiting list is not empty.          

\subsection{Computational Complexity Analysis} 
Interference graph transformation procedure given in Algorithm \ref{Alg1} takes $O(K+N)$ computational complexity. To observe the number of fixed resources $R$ for elaborated FAPs as given in Algorithm \ref{Algb} has to take $O(K)$. This is because, it goes to all the FAPs to decide the reserved resources. In order to estimate the computational complexity of Algorithm \ref{Alg2}, we first compute the computational complexity of steady FAPs then for elaborated FAPs. As every time whenever the BS sends message to a FAP, it removes that FAP from its candidate list (as shown in line 4 of Algorithm \ref{Alg2}). At the end, for each resource, the set of interference free FAP assigned with the same resource. As we have $N$ PRBs, each time BS selects FAP by looking into the maximum weighted independent graph. Thus, the computational time complexity to allocate the valid resources to each steady FAP be $O(KN\rho)$ where, $\rho$ is time complexity for finding MWIS in a graph. When, the Algorithm \ref{Alg2} assigns the resource to a elaborated FAPs (lines 26-36), all the elaborated FAPs needs to be traverse and that may take $O(K)$ computational time complexity. Thus, the proposed resource allocation method takes total $O(NK^2\rho)$ computational time complexity.

\subsection{Performance Bound of Proposed Algorithm} 
In this subsection, we compute the number of PRB bound over the propagation latency of our proposed scheme. Let $\mathbb{H}$ and $\mathbb{L}$ be the number of higher priority and lower priority IoT devices associated with each FAP, respectively. Assuming that each IoT device has to send 504 bit of data within propagation delay of 0.5 ms and BS operates at 64 QAM, then there is a need of at least one PRB for an IoT device \cite{hatoum2014cluster}, \cite{SMC}. If $Q$ and $\Delta$ represent clique and degree of interference graph $G(V, E)$, respectively. Then we bound the number of required PRBs in the network as follows: 
\begin{lem}\label{lem1}
If propagation latency of IoT devices carrying 504 bits of data is bounded by the 0.5 ms then our propose algorithm bounds over the required number of PRBs $\Gamma$ in the network as $\mathbb{H}Q\leq \Gamma \leq (\mathbb{H}+\mathbb{L})(\Delta+1)$. \end{lem}

 
\begin{proof}
As the interference graph has a clique of size $Q$ and each node has minimum PRB demand of $\mathbb{H}$, consequently to fulfill the minimum PRB demand of each FAP, the number of required PRB $\Gamma$ must follow the lower bound constraint such as $Q\mathbb{H}\leq \Gamma$. On the other hand, to prove the upper bound, we borrow the degree concept of graph coloring procedure. Based on graph coloring concept we can say the number of required colors are bounded as $(\Delta+1)$ if each node is labeled with a single color \cite{welsh1967upper}. Accordingly, we can write the upper bound over the required number of PRBs, when each node is assigned with a maximum of $(\mathbb{H}+\mathbb{L})$ number of PRBs. Thus, we can write the upper bound for required number of PRBs such as $\Gamma \leq (\mathbb{H}+\mathbb{L})(\Delta+1)$. Combining both the upper and lower bounds, we obtain the result stated in Lemma \ref{lem1}.            
\end{proof}


\section{Performance Study}\label{Sec: Performance Study} 
In this section, we evaluate our proposed method based on the following environments. We have assumed that the set of FAPs and IoT devices are deployed randomly in the network area of $500$ m x $500$ m underlying a cellular BS. The bandwidth is considered as 20 MHz accordingly, the maximum number of available PRBs found as 100 \cite{chiu2016ultra}. We have assumed that FAP with a radius of 20 m, randomly deployed in the network. Transmit power of IoT devices is set to 25 dBm. A distance between IoT and FAP is considered between 10 - 15 m. Like \cite{hatoum2014cluster, zirwas2018sub} we also considered that a PRB can carry 504 bits of data at 64 QAP modulation scheme. The network model is considered as the same \cite{gu2018joint, hasan2014distributed}. The noise power -114 dBm. The required latency, data size, and corresponding CPU cycles are determined by specific device types. The total latency is considered as the sum of uplink latency, execution latency, and downlink latency. The total latency requirement of each IoT device is randomly distributed within [0.1 - 5] minutes. We have considered CPU frequency of FAP processor as 1.4 GHz and computational complexity of the task as 10 computation cycles/bit \cite{vu2018joint}. For the sake of simplicity we have assumed the response time i.e., downlink latency as a random variable $\gamma_{k,i}(\boldsymbol{y}_{k,i})=\delta t$, $\delta t \in [0,1]$ for any IoT device \cite{gu2018joint}.

\subsection{Utility as a Surrogate to Task Throughput}
In our results, we assume that a task generated by an IoT device will need one PRB to get it executed. Therefore, the total number of tasks served by our proposed algorithm (a.k.a. task throughput) is equal to the total number of PRBs assigned by the BS. In other words, the fraction of assigned PRBs to the total PRBs demanded by IoT device is a monotonically increasing function of the total number of tasks served by the proposed algorithm. As a result, we evaluate the performance of the proposed algorithm by determining utility based resource allocation in three case scenarios of an IoT framework. In the field of network communication, the utility of resources determines the ratio between the amount of the PRBs that are allocated to IoT devices to number of required PRBS of that device i.e., $Utility_i= \displaystyle \frac{1}{d_i}\sum_{n \in \mathbb{N}} y_{k,i}^n$, for any $F_k \in \mathbb{F}$. In the below, we define the utility of higher priority and lower priority IoT devices based on formulated Problem (P1).
 \begin{equation}\label{eq17}
Utility_i =
\begin{cases}
1, & \text{if $w_i=1$ } 
\\
\leq 1, & \text{otherwise}.
\end{cases}
\end{equation}
  \begin{figure}[h!]
\vspace{-3ex}
	\centering
	\includegraphics[height=3.3cm, width=9.5cm]{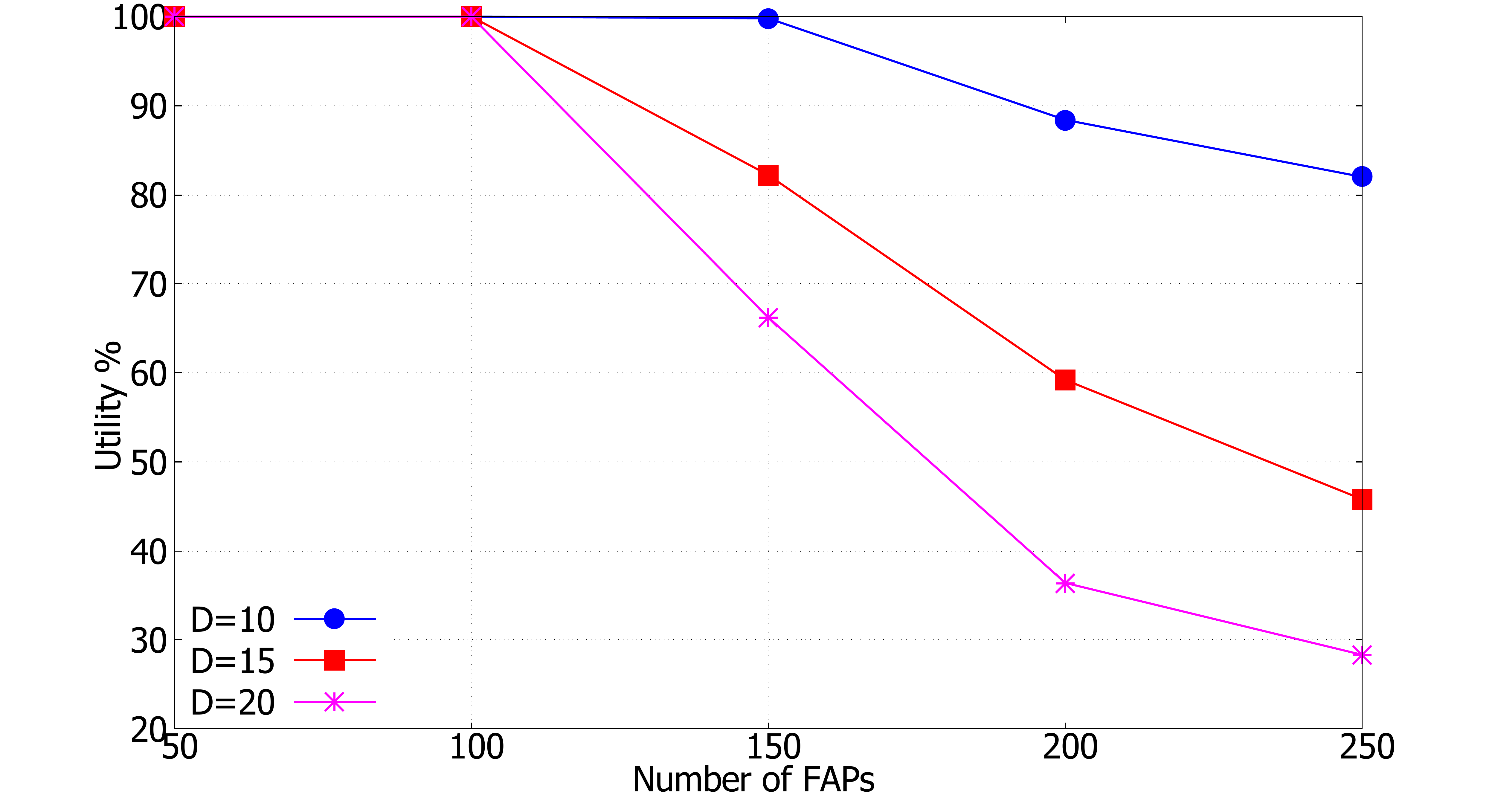}
	\vspace{-3ex}
	\caption{Comparison between number of IoT devices and utility value.}
	\label{re1}
\end{figure}
{\textit {A) Impact of IoT devices on utility:}} For an experimental analysis, a set of FAPs ranging from 50 to 250 was chosen. For each set of FAPs, the proposed algorithm is estimated to obtain utility by setting the maximum demand of PRBs i.e., $D$ from 10 to 20. The utility is calculated as an average value. Comparison of utility with respect to FAP is shown in Fig.\ref{re1}. From the graph, we can see that, for a specific demand of resources when there is an increase in total PRB demand, there is a gradual decrease in utility with a sudden fluctuation as number of FAPs approach to 250. This is because as the number of IoT devices increase, there is a throttle for resources in the network leading to a performance decline of the network and hence lower utility value. Additionally, as the demand for resources increase for the same number of FAPs, utility of network resources decrease due to scarcity of resources to meet a higher demand. 

{\textit {B) Impact of interference link density on utility :}} Fig. \ref{re2} demonstrates the graph for utility v/s interference of link density. Link density is defined as a ratio between actual edges in interference graph to maximum possible edges. This is performed to study impact of network interference on utility of resources among network devices. Utility is different for different link density values. 
\begin{figure}
	\centering
	\includegraphics[height=3.3cm, width=9.5cm]{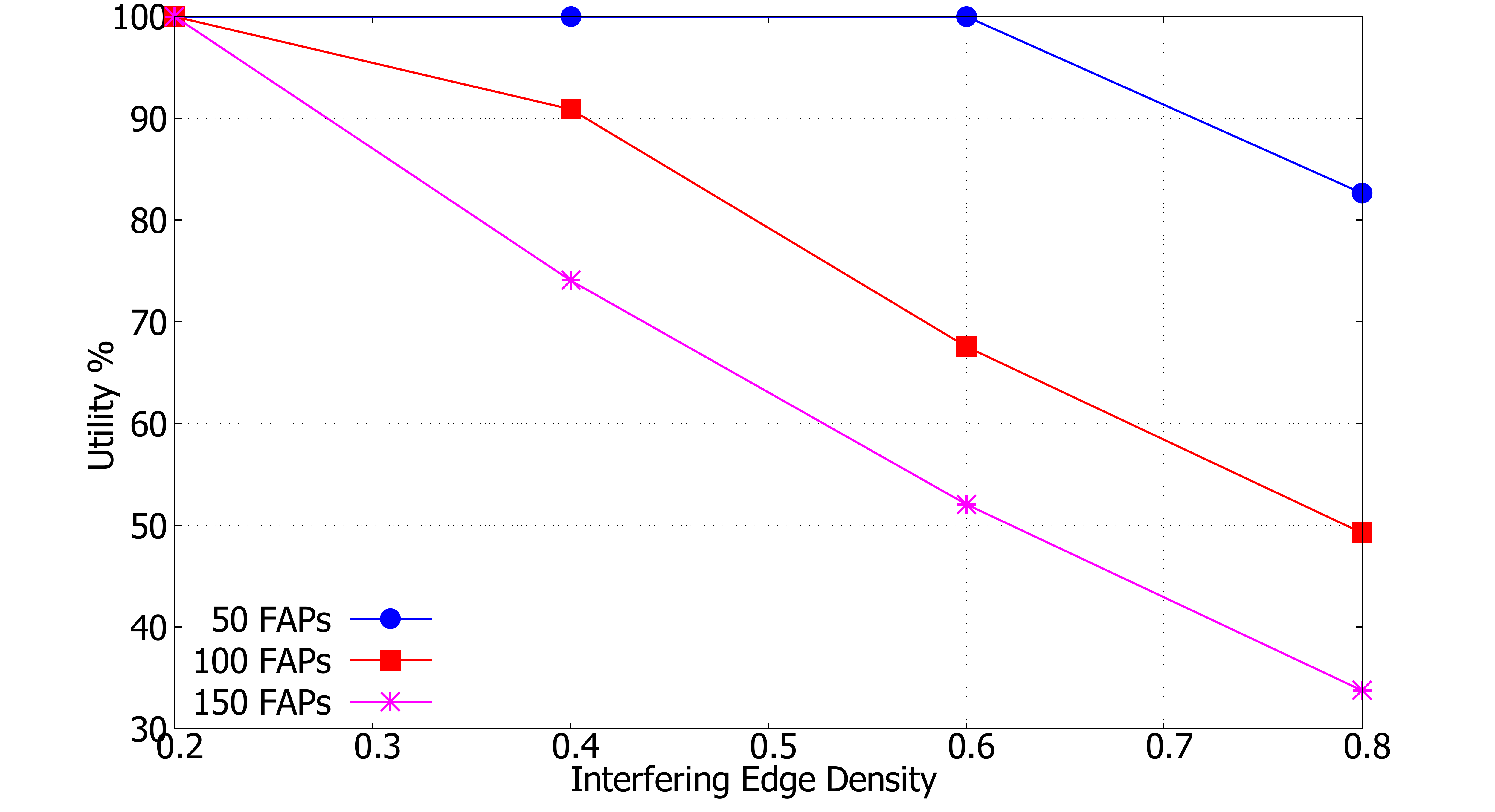}
	\vspace{-3ex}
	\caption{Comparison between interference link densities and utility value.}
	\vspace{-3ex}
	\label{re2}
\end{figure}
From the result, we can observe that as link density increases, resource utility starts decreasing gradually. This is because for a given set of FAPs when there is a higher link density there is stronger interference in wireless link that lead to lower resources re-usability. This causes decline in network resource utility of devices.
\begin{figure}[h!]
	\centering
	\includegraphics[height=3.3cm, width=9.5cm]{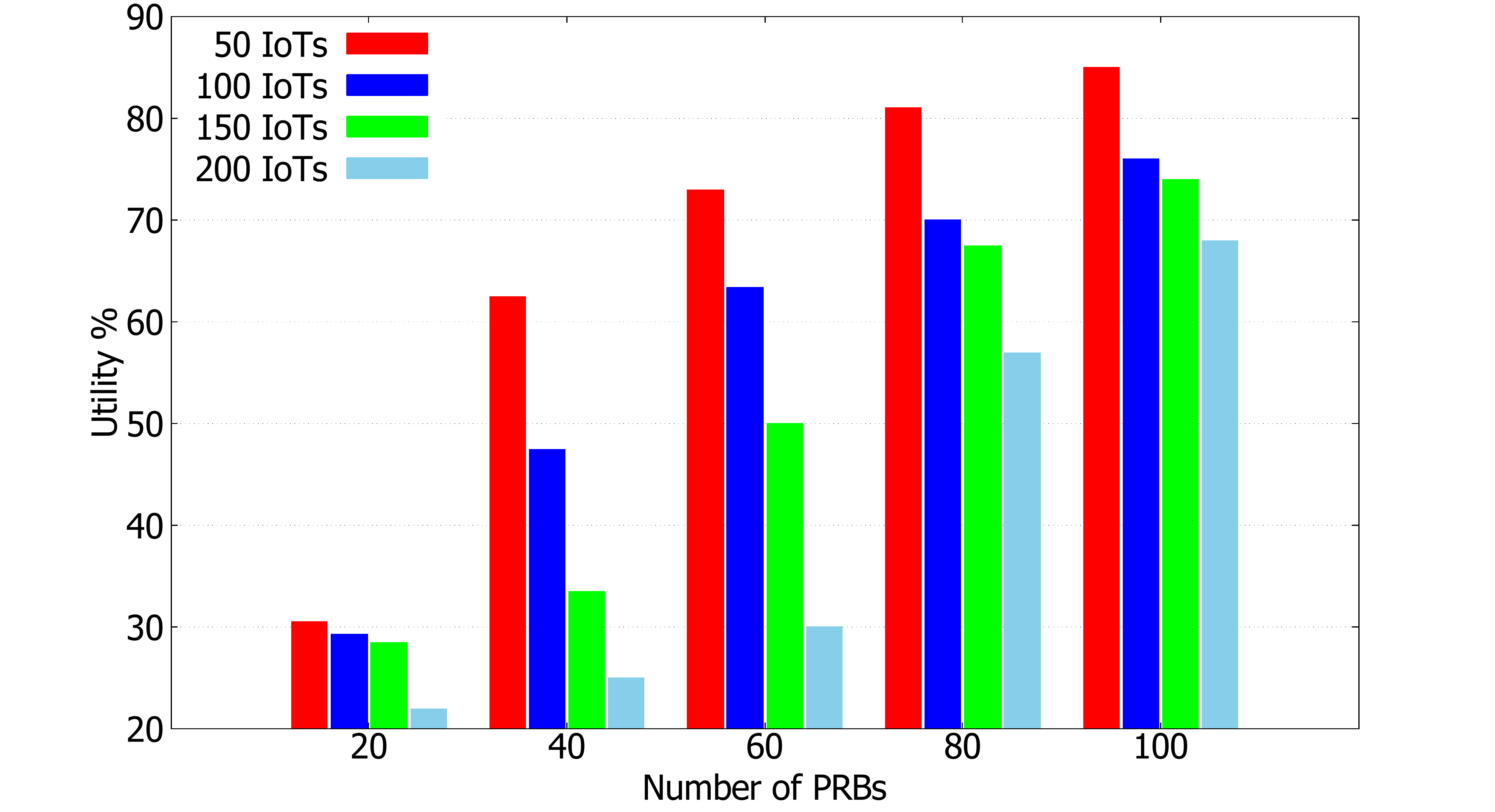}
	\caption{Comparison between utility and number of resources.}
	\label{re3}
\end{figure}

{\textit {C) Impact of number of allocated resources on utility:}} In Fig. \ref{re3}, we compare the utilities along with an increment of the number of allocated PRBs. For a given set of IoT devices, as the number of allocated resources in the network increase, there is a sharp increase in the utility as now more resources can be assigned to each IoT devices. If the number of allocated resources are small, the devices will not be serviced as per their demand due to interference constraint and hence, smaller utility. However, if the total resources in the network are sufficient, the effective resources allocated to each device is higher, which results in higher utility for the network. Additionally, it is evident from the graph that as the number of IoT devices increase the utility will decrease for the given amount of resources. This is because for a specific amount of allocated resources, as the number of devices in the network increase, more subgroups of non-interfering devices will be formed that can lead to a division of resources into more number of portions. Due to this, some devices may fall short of meeting their minimum resource requirement which causes a drop in the utility of the network resources.  

\subsection{Latency Evaluation}
In this sub-section, we present a prototype of the proposed architecture to compare the performance of the proposed scheme. Further, we compare the total latency of the developed prototype with the obtained numerical results. We assumed that two PRBs are dedicated to each device in order to transmit the data to respective FAP. 
\subsubsection{Prototype Specification} To create a prototype of the proposed model, we use a laptop, an android mobile phone, and one raspberry pi as IoT devices sending data over Wi-Fi to another raspberry pi which is working as a FAP in our model. \begin{figure}[h!]
\vspace{-2ex}
	\centering
	\includegraphics[height=3.4cm, width=9.0cm]{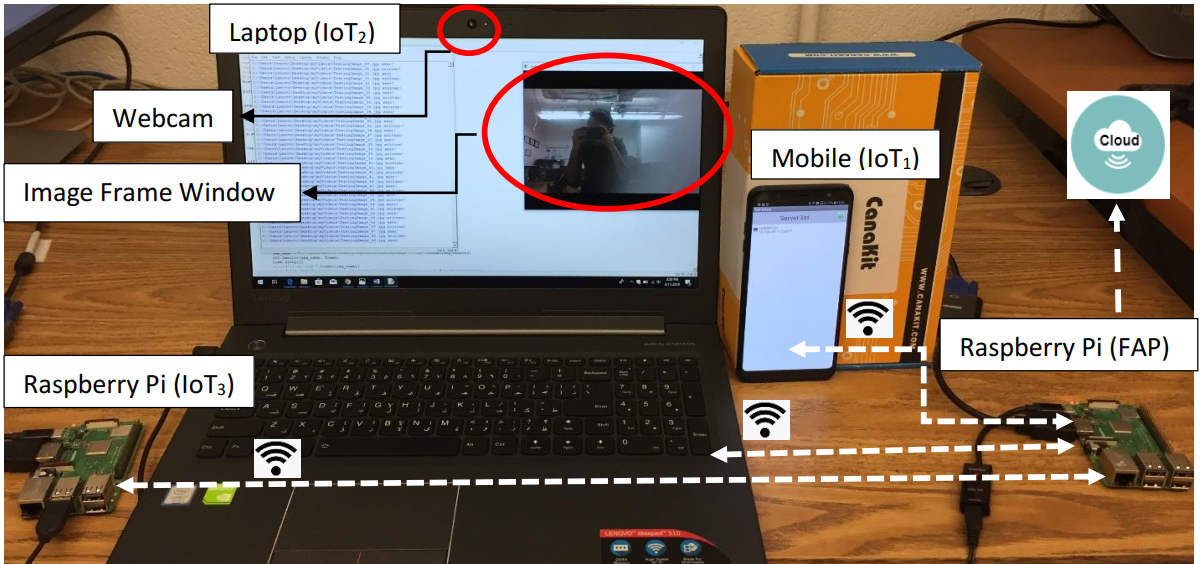}
	\vspace{-3ex}
	\caption{Prototype setup with mobile, laptop, raspberry pi and cloud with their communication medium.}
	\vspace{-1ex}
	\label{Type}
\end{figure}
The specification of each device used in the prototype model is shown in Table \ref{table11}. The camera enabled laptop takes photos in every interval of 1 ms and forwards the images to the central FAP. Another raspberry pi working as IoT device continuously sends stored images to the raspberry pi working as FAP. Additionally, the android mobile phone sends some continuous text messages through a client application to the FAP raspberry pi over the Wi-Fi network. On the FAP side, upon receiving the data from all three IoT devices, the raspberry pi FAP performs local data processing and once the execution is completed it sends response messages to each respective IoT device.  
\begin{table}[h!]
\vspace{-4ex}
	\begin{center}
		\caption{Hardware specification of prototype.}
		\label{table11}
		\begin{tabular}{|c|c|} 
		\hline
			\textbf{Devices} & \textbf{Specification} \\			
			\hline
			Raspberry Pi  &  \begin{tabular}[c]{@{}l@{}}
			\text{\textbf{Model:} Raspberry Pi 3 B+}, \\
			\text{\textbf{SOC:} Broadcom BCM2837B0,} \\
			\text{Cortex-A53 (ARMv8) 64-bit SoC.}\\
			\text{\textbf{CPU:} 1.4GHz 64-bit quad-core,} \\
			\text{ARM Cortex-A53 CPU,} \\
			\text{\textbf{RAM:} 1GB SDRAM,} \\
			\text{\textbf{WiFi:} Dual-band 802.11ac wireless LAN,} \\
			\text{(2.4GHz and 5GHz ) Bluetooth 4.2.} \\
			\text{\textbf{Power:}  5V/2.5A DC power input.} 		\end{tabular}\\
			\hline 
            Laptop  &  \begin{tabular}[c]{@{}l@{}}
			\text{\textbf{Model:} Lenovo Ideapad 510.} \\
			\text{\textbf{Processor:} Core i5 7th Gen.} \\
			\text{\textbf{CPU:} 2.5 Ghz.} \\
			\text{\textbf{RAM:} 2GB SDRAM,} \\
			\text{\textbf{WiFi:}802.11 a/b/g/n/ac, Bluetooth 4.2} \\
			\text{\textbf{Web-Cam:} Yes, recording at 720p HD,} \\
			\text{\textbf{Battery:}   Cell Li-Ion.}
		\end{tabular} \\
\hline 
Mobile &   \begin{tabular}[c]{@{}l@{}}
			\text{\textbf{Model:} Samsung Galaxy A6 +.} \\
			\text{\textbf{OS:} Android 8.0 (Oreo).} \\
			\text{\textbf{CPU:} Octa-Core 1.8GHz Cortex A53.} \\
			\text{\textbf{RAM:} 4GB,} \\
			\text{\textbf{WiFi:}8802.11 a/b/g/n, WiFi Direct, } \\
				\text{hotspot and Bluetooth 4.2} \\
			\text{\textbf{Network technology:} GSM/HSPA/LTE.}
		\end{tabular} \\
\hline
		\end{tabular}
	\end{center}
	\vspace{-4ex}
\end{table}
\begin{figure}[h!]
\vspace{-1ex}
	\centering
	\includegraphics[height=3.3cm, width=9.5cm]{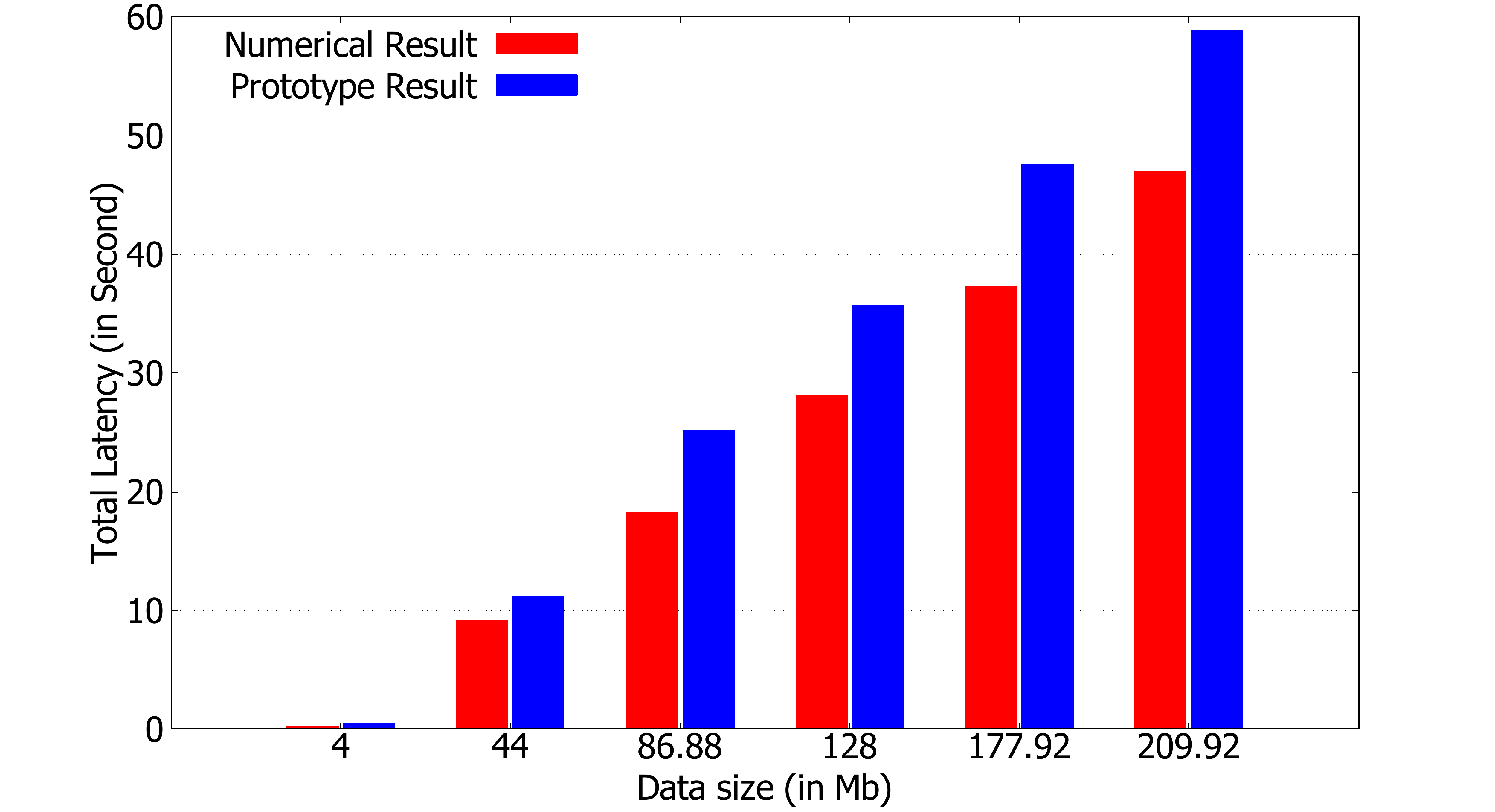}
	\vspace{-3ex}
	\caption{Comparison of total delay and data size on developed prototype and numerical result.}
	\label{Comp}
	\vspace{-1ex}
\end{figure} 
We estimate the proposed scheme on different size of data varying from 4 Mb to 209.92 Mb by assuming the task deadline as 1 minute. With the increase of data size, the total time consumed to execute the task also increases. As from Fig. \ref{Comp} we can observe that the total consumed latency in numerical result and prototype model is almost the same. However, with the increase in data size, the total latency consumed by the prototype model is increasing. The main reason for this difference is that, laptop and smart-phone used as IoT devices are not dedicated devices and these devices are running several other applications in parallel while in use as an IoT device in the prototype model. 

\subsection{Joint PRB and Latency Evaluation}
 To the best of our knowledge, most of the existing work for resource allocation in the IoT-enabled network did not consider the priority of tasks execution keeping limited PRB and latency constraint into consideration. So, in order to compare our proposed scheme with the existing work \cite{vu2018joint}, we draw a random topology of 20-100 number of IoT devices getting services from 5 FAPs. We assumed that each IoT device will send 1 Mb of data to get executed at FAP. We keep the balance coefficient $\alpha=0$ for the work \cite{vu2018joint} in order to compare the latency with our proposed scheme. We can conclude from Fig. \ref{ExCom}, with the increase in the number of PRBs in the network total latency, is minimized. The reason from this observation is that when the more number of PRBs are allocated for the same task, the total achievable data rate goes high and propagation latency becomes low, respectively. Our proposed scheme gives lower latency compared to the existing work. The reason is that unlike the existing work our proposed scheme always try to maximize the re-usability of PRBs while avoiding interference among FAPs and this phenomenon improves the achievable data rate between FAP and IoT device consequently, our proposed scheme results lower latency in the networks.      

\begin{figure}[h!]
\vspace{-3ex}
	\centering
	\includegraphics[height=3.3cm, width=9.0cm]{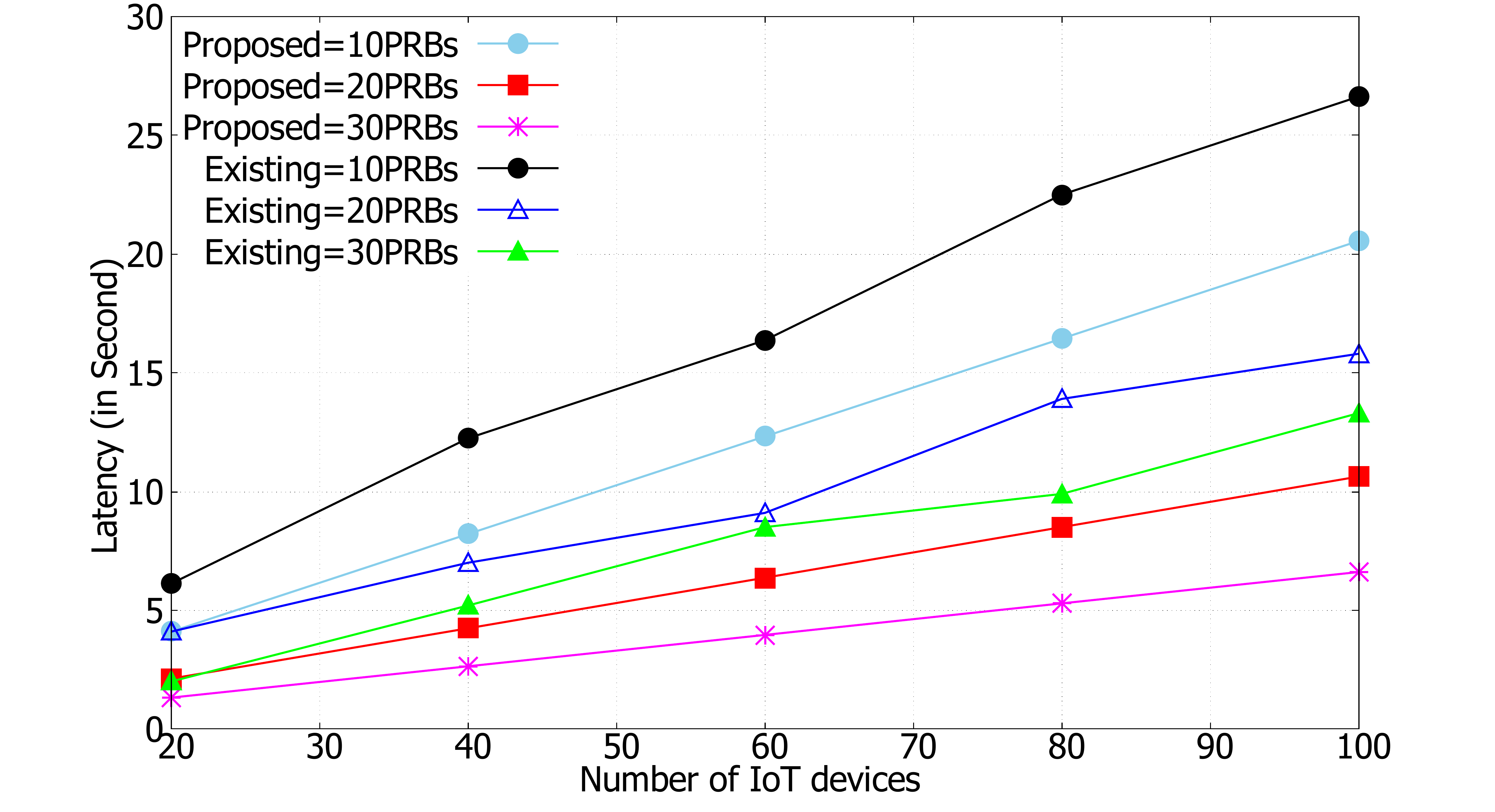}
	\vspace{-3ex}
	\caption{Joint PRB and latency comparison with existing work.}
	\vspace{-3ex}
	\label{ExCom}
\end{figure}

\section{Conclusion}\label{Sec: Conclusion} 
In this work, we proposed a novel framework for 5G networks where the BS identifies appropriate pairs of IoT devices and FAPs and allocate necessary resources to maximize the total number of tasks served. We proposed a graph-coloring based algorithm to solve it in a computationally tractable manner. Through numerical result and prototype model, we have demonstrated the effect of different parameters over the utility and latency of IoT devices in different environments. In the future, the proposed scheme can be applied to specific applications via considering specific applications by modifying the set of constraints accordingly.

\bibliographystyle{ieeetr}
\bibliography{IEEEabrv,bibl}

\end{document}